\documentclass[11pt]{article}
\usepackage[english]{babel}

\usepackage[letterpaper,top=2cm,bottom=2cm,left=3cm,right=3cm,marginparwidth=1.75cm]{geometry}

\usepackage{amsmath}\usepackage{amsthm, amssymb}
\usepackage{graphicx}
\usepackage[colorlinks=true, allcolors=blue]{hyperref}

\title{Statistical characterization of the chordal product determinant of Grassmannian codes}
\author{Javier Álvarez-Vizoso, Carlos Beltrán, Diego Cuevas, Ignacio Santamaría,\\ V{\' i}t Tu{\v c}ek and Gunnar Peters}

\newcommand{\Addresses}{{
  \bigskip
  \footnotesize

\textsc{Department of Communications Engineering, Universidad de Cantabria,  Santander, Spain}\par\nopagebreak
  \textit{E-mail addresses}: \texttt{javier.alvarezvizoso@unican.es},  \texttt{diego.cuevas@unican.es},\texttt{i.santamaria@unican.es}
  \medskip

\textsc{Department of Mathematics, Statistics and Computing, Universidad de Cantabria, Santander, Spain}\par\nopagebreak
  \textit{E-mail address}: \texttt{carlos.beltran@unican.es}

  \medskip

 \textsc{Department of Wireless Algorithms, Huawei Technologies, Kista, Sweden}\par\nopagebreak
  \textit{E-mail addresses}: \texttt{vit.tucek@huawei.com},  \texttt{gunnar.peters@huawei.com}

}}

\begin{document}
\maketitle

\begin{abstract}
  We consider the chordal product determinant, a measure of the distance between two subspaces of the same dimension. In information theory, collections of elements in the complex Grassmannian are searched with the property that their pairwise chordal products are as large as possible. We characterize this function from an statistical perspective, which allows us to obtain bounds for the minimal chordal product and related energy of such collections. 
\end{abstract}

\def\gg{\left|g\right|^2}
\def\hg{\left|h\right|^2}
\def\dg{\left|d\right|^2}

\def\C{{\bf C}}
\def\N{{\bf N}}
\def\E{{\mathcal E}}
\def\S{{\bf S}}
\def\X{{\bf X}}
\def\x{{\bf x}}
\def\z{{\bf z}}
\def\f{{\bf f}}
\def\s{{\bf s}}
\def\e{{\bf e}}
\def\y{{\bf y}}
\def\n{{\bf n}}
\def\T{{\mathcal T}}
\def\L{{\mathscr L}}
\def\G{{\bf G}}
\def\U{{\bf U}}
\def\H{{\bf H}}
\def\V{{\bf V}}
\def\u{{\bf u}}
\def\v{{\bf v}}
\def\h{{\bf h}}
\def\I{{\bf I}}
\def\R{{\bf R}}
\def\Q{{\bf Q}}
\def\W{{\bf W}}
\def\F{{\bf F}}
\def\Y{{\bf Y}}
\def\H{{\bf H}}
\def\X{{\bf X}}
\def\A{{\bf A}}
\def\B{{\bf B}}
\def\D{{\bf D}}
\def\Z{{\bf Z}}
\def\Y{{\bf Y}}
\def\tr{\operatorname{tr}}
\def\det{\operatorname{det}}
\def\rank{\operatorname{rank}}
\def\diag{\operatorname{diag}}
\def\blkdiag{\operatorname{blkdiag}}
\def\det{\operatorname{det}}
\def\etr{\operatorname{etr}}
\def\Exp{\operatorname{E}}
\def\trace{\operatorname{tr}}
\def\argmin{\operatorname{argmin}}
\def\argmax{\operatorname{argmax}}

\def\SigmaB{\boldsymbol{\Sigma}}
\def\PsiB{\boldsymbol{\Psi}}

\def\Gras{\mathbb{G}r(M,{\mathbb{C}}^{T})}
\def\St{\mathbb{S}t(M,{\mathbb{C}}^{T})}
\def\dc{\mathrm{d}_{\mathrm{ch}}}
\def\Jac{\mathrm{Jac}\;}
\def\Grasb{\tilde{\mathbb{G}r}(M,{\mathbb{C}}^{T})}
\newcommand{\Id}{\mathrm{Id}}
\def\CU{{\mathcal U}}

\newcommand{\nachosays}[1] {{\bf{\textcolor{blue}{[Nacho SAYS: #1]}}}}
\newcommand{\carlossays}[1] {{\bf{\textcolor{red}{[Carlos SAYS: #1]}}}}
\newcommand{\javiersays}[1] {{\bf{\color{Green}{[Javier SAYS: #1]} }}}
\newcommand{\diegosays}[1] {{\bf{\color{Orange}{[Diego SAYS: #1]} }}}

\ifdefined\begintheorem\else
\newtheorem{thm}{Theorem}
\fi
\newtheorem{prop}{Proposition}
\newtheorem{cor}{Corollary}
\newtheorem{lem}{Lemma}
\newtheorem{rmk}{Remark}
\newtheorem{defi}{Definition}
\newtheorem{exmpl}{Example}

\def\proof{{\noindent\sc Proof. \quad}}
\newcommand{\proofof}[1]{{\noindent\sc Proof of #1. \quad}}
\def\eproof{{\mbox{}\hfill\qed}\medskip}

\def\qedsymbol{\vbox{\hrule\hbox{%
                     \vrule height1.3ex\hskip0.8ex\vrule}\hrule}}
\def\endproof{\qquad\qedsymbol\medskip\par}
\def\myendproof{\qquad\qedsymbol}

\section{Introduction and statement of the main results}
Let $T\geq 2M$ be two positive integers and consider the complex Grassmannian $\Gras$, i.e. the space of $M$--dimensional complex vector subspaces of $\mathbb C^T$.  Finite collections of points (also called {\em codes} or {\em packings}) in $\Gras$ with different desired separation properties have been investigated by several authors (in Section \ref{sec:history} we describe some relevant references). The most frequent criterium for ``well--separated'' codes is the maximization of the minimal mutual squared chordal distance, which is the sum the squared sines of the principal angles of two subspaces. However, following \cite{Hochwald00,Varanasi02,CuevasTCOM} (see also Section \ref{sec:aplicacion} below), a more relevant measure for its application to information theory is given by the {\em chordal product energy}, related to the product of the squared sines of the principal angles, which justifies its name.
Given a code $[\X_1],\ldots,[\X_k]\in\Gras$, its chordal product energy with parameter $N$ is
  \begin{equation}\label{eq:unionbound2}
\mathcal E(\X_1,\ldots,\X_K)= \sum_{i\neq j}\det ( \I_M -\X_i^H\X_j \X_j^H \X_i) ^{-N},
\end{equation}
where $\I_M$ is the identity matrix and we have chosen representatives $\X_i$ of each point $[\X_i]$ satisfying $\X_i^H\X_i=\I_M$. Note that the energy is well defined in the sense that it does not change if other representatives with that property are chosen.
Recall that the SVD
of $\X_i^H\X_j$, e.g. $\U\D\V^H$, is given in terms of the cosines of the principal angles between the subspaces $[\X_i]$ and $[\X_j]$, $\cos \theta_1, \ldots, \cos\theta_M$, cf. \cite{HanTIT06}, so
\begin{equation}
    \det \left( \I_M - \X_i^H\X_j \X_j^H\X_i \right) =\det \left( \I_M - {\bf D}^2 \right)= \prod_{i=1}^M \sin^2 \theta_i,
\end{equation} 
while the squared chordal distance between the two subspaces $[\X_i]$ and $[\X_j]$ is given by $\sum_{i=1}^M \sin^2 \theta_i$.
 The sum in \eqref{eq:unionbound2} is a pairwise interaction energy in the spirit of the well--studied Riesz or logarithmic energies of importance in Potential Theory (see \cite{SaffBook} for a complete monograph dedicated to energy minimization in the sphere and other spaces). We refer to the function (again, choosing representatives $\A$ and $\B$ such that $\A^H\A=\B^H\B=\I_M$)
\begin{equation}
    [\A],[\B]\in\Gras\mapsto\det ( \I_M -\A^H\B \B^H \A) =\det ( \I_M -\B^H\A \A^H \B) ,
    \label{eq:coherenceDet}
\end{equation}
as the {\em chordal product determinant} or, simply, the chordal product, and note that is {\em not} a metric in $\Gras$, for it may happen that $[\A]\neq[\B]$ and yet $\det ( \I_M -\A^H\B \B^H \A)=0$, if the intersection of $[\A]$ and $[\B]$ is nontrivial.

In this paper we perform the first theoretical study of the chordal product energy, for numerical results, see \cite{CuevasTCOM} and references therein. We start describing the context where the problem arises, following \cite{Hochwald00}.

\subsection{The importance of Grassmannian codes in information theory}\label{sec:aplicacion}
Consider a {\em transmitter}, i.e. some device that is able to send a signal, which is indeed a collection of numbers ordered in a complex $T\times M$ matrix $\X$. Physically, this corresponds to the setting where the transmitter has $M$ antennas and there is a total amount of $T$ time slots where the communication channel is assumed to be constant (i.e. the contour conditions of the communication are considered constant during the time that these $TM$ numbers are sent). The {\em receiver} is another device, that we consider equipped with $N$ antennas, and the signal it receives is 
\begin{equation*}
    \mathbf{Y} = \mathbf{X} \mathbf{H} + \sqrt{\frac{M}{T \rho}} \mathbf{W},
\end{equation*}
where $\mathbf{H}$ is an unknown $M\times N$ matrix (termed {\em the channel}), $\mathbf{W}$ describes the noise and $\rho$, called the signal-to-noise-ratio (SNR), measures the magnitude of the signal against the noise.
\subsubsection{The zero--noise case}
Since $\mathbf{H}$ is unknown (it is common to assume that it has random complex  Gaussian entries), even in the event that $\mathbf W=0$ the receiver cannot recover the whole matrix $\X$: 
\begin{itemize}
    \item If two matrices $\X_1$ and $\X_2$ have the same column span, then one can easily find an full--rank matrix $\mathbf H$ such that $\mathbf{X}_1 \mathbf{H}=\mathbf{X}_2 \mathbf{H}$, hence the receiver just cannot distinguish which of these two matrices was the original signal. 
    \item On the other hand, if two matrices $\X_1$ and $\X_2$ have the property that the intersection of the column span of $\X_1$ and $\X_2$ is trivial, the the receiver can easily distinguish if a given matrix $\Y$ has been constructed by $\X_1\mathbf{H}$ or by $\X_2\mathbf{H}$: if the column span of $\Y$ intersected with the column span of $\X_1$ (resp. $\X_2$) is nontrivial, then $\X_1$ (resp. $\X_2$) was sent.
\end{itemize} 
Summarizing, if a previously agreed code of possible signals $[\X_1],\ldots,[\X_K]\in\Gras$ is fixed with the property that the column spans of $\X_i$ and $\X_j$ have trivial intersection for $i\neq j$, the receiver will be able to recover, at least in the zero--noise scenario, {\em the element of the Grassmannian represented by the sent signal, but not the concrete representative of that element.} Hence, collections of points in $\Gras$ are searched with that property. 
\subsubsection{The general case}
In the more realistic context of the presence of non--zero noise, the analysis is quite more involved since there is always a non--zero probability of error in the detection procedure. The pioneer work \cite{Hochwald00} showed that, in order to recover the element $\X_i$ of $\Gras$ just by knowing $\Y$, the optimal method is to use the so called maximum--likelihood decoder:
\[
i=\argmax_{j=1,\ldots,K}\trace{(\Y^H\X_j\X_j^H\Y)}=
\argmax_{j=1,\ldots,K}\trace{(\X_j^H\Y\Y^H\X_j)},
\]
where $\cdot^H$ holds for Hermitian conjugate. Then, \cite{Varanasi01} showed that if only $2$ codewords are permitted, i.e. if $K=2$, and assuming that the entries of $\mathbf{H}$ and $\mathbf{W}$ are complex Gaussian $\mathcal{N}(0,1)$ numbers, then the probability $P_e(\X_1, \X_2,\rho)$ of erroneously decoding $\X_1$ if $\X_2$ was sent can be given by a (quite complicated) formula involving the residues of a certain rational function. Luckily, the asymptotic expansion of this {\em Pairwise Error Probability} (PEP) in the case $\rho\to\infty$, called the high-SNR asymptotic analysis, admits a much more concise expression, see  \cite{Varanasi02,CuevasTCOM}:
\begin{equation}
    P_e(\X_1, \X_2,\rho) \approx C \rho^{-NM}  \det ( \I_M -\X_1^H\X_2 \X_2^H \X_1) ^{-N},\quad \rho\to\infty,
    \label{eq:PEPbis}
\end{equation}
where $C= \frac12\left( \frac{4M}{T}\right)^{NM} \frac{(2NM-1)!!}{(2NM)!!)}$, it is assumed that any two distinct points have trivial intersection as linear subspaces, and the representatives $\X_i$ of each $[\X_i]$ are such that $\X_i^H\X_i=\I_M$. If we have $K$ elements $[\X_1],\ldots,[\X_K]$ in the code of possible signals and we assume that we send one of them at random, all with equal probability $1/K$, then the total probability of erroneously decoding a signal is bounded above by
\begin{equation}\label{eq:unionbound}
\frac1K\sum_{i\neq j}P_e(\X_i, \X_j,\rho)
\approx
\frac{C}{K} \rho^{-NM} \sum_{i\neq j}\det ( \I_M -\X_i^H\X_j \X_j^H \X_i) ^{-N}.
\end{equation}
The determinant in \eqref{eq:PEPbis} is the chordal product \eqref{eq:coherenceDet} and the sum in the right--hand side in \eqref{eq:unionbound} is the energy \eqref{eq:unionbound2}.

\subsubsection{Criteria for the design of Grassmannian codes}

It follows from the previous discussion that reasonable criteria for the design of a code $[\X_1],\ldots,[\X_K]$ would be to maximize the pairwise chordal product  \eqref{eq:coherenceDet},  or to minimize the chordal product  energy \eqref{eq:unionbound2}. In \cite{CuevasTCOM} these approaches are considered, numerically showing that the obtained codes are very well suited for their use in non--coherent communications, with a slight advantage in the use of the chordal product energy.  Yet, little or no theory exists about the behavior of the optimal pairwise chordal product or energy. The main purpose of this paper is to put the basis for the study of this question.

\subsection{Main results of the paper}
 We will start our study by computing the moments of the chordal product when $[\B]$ is fixed and $[\A]$ is chosen at random uniformly in $\Gras$, w.r.t. the unique, standard rotation--invariant probability measure. This yields a complete statistical characterization of the chordal product as a product of beta--distributed random variables:
\begin{thm}\label{th:expected}
  Assume that $T\geq 2M$. Let $p\in(2M-T-1,\infty)$ (notice that $p$ may be negative and/or noninteger). Let $[\B] \in\Gras$ be any fixed element and let $[\A] \in\Gras$ be uniformly distributed on the Grassmannian. Then, the $p$--th moment of $\det(\I_M-\B^H \A\A^H \B)$ is:
\begin{equation}\label{eq:detmoment}
  {\rm E}_{\A}[\det(\I_M-\B^H \A\A^H \B)^p] =
  \prod_{m=1}^M\frac{\Gamma(T-m+1)\Gamma(T+p-m-M+1)}{\Gamma(T-m-M+1)\Gamma(T+p-m+1)},
\end{equation}
where $\Gamma(\cdot)$ is Euler's Gamma function. Moreover, $\det(\I_M-\B^H \A\A^H \B)$ is distributed as the product of $M$ independent beta random variables, $z_m$, with parameters $\alpha_m = T-M+1-m$ and $\beta_m = M$, $m=1,\ldots,M$, i.e.
\begin{equation}
\det(\I_M-\B^H \A\A^H \B) \sim \prod_{m=1}^M z_m, \quad z_m \sim {\rm Beta} (T-M+1-m,M).
\end{equation}
\end{thm}
An immediate consequence is that, at least for moderate values of $N$, we can upper bound the energy \eqref{eq:unionbound2} and hence the probability of error \eqref{eq:unionbound} of random codes $[\X_1],\ldots,[\X_K]$ when they are all independently and uniformly distributed:
\begin{cor}\label{cor:optimalrandom}
Assume that $N\leq T-2M$. For i.i.d. chosen $[\X_1],\ldots,[\X_K]$, the expected value of the chordal product energy \eqref{eq:unionbound2} is
\[
K(K-1)\prod_{m=1}^M\frac{(T-m)!(T-N-m-M)!}{(T-m-M)!(T-N-m)!},
\]
In particular, there exists a code such that the union bound \eqref{eq:unionbound} is at most:
\[
C(K-1)\rho^{-NM}\prod_{m=1}^M\frac{(T-m)!(T-N-m-M)!}{(T-m-M)!(T-N-m)!},
\]
where $C= \frac12\left( \frac{4M}{T}\right)^{NM} \frac{(2NM-1)!!}{(2NM)!!)}$.
\end{cor}
In Section \ref{sec:pdfcdf} we use Theorem \ref{th:expected} to compute exactly the probability density function and the cumulative density function of the random variable $\det(\I_M-\B^H\A\A^H\B)$ in the same hypotheses of the theorem. The expressions we get are exact and can be obtained in closed form for any fixed value of $M$. A reduced version of that result for $M=2$ is now shown:
\begin{cor}\label{cor:M2} Fix any $[\B]\in \mathbb{G}r(2,\mathbb C^T)$ with $T\geq4$. The probability that a randomly chosen $[\A]\in\mathbb{G}r(2,\mathbb C^T)$ satisfies $\det(\I_2-\B^H\A\A^H\B)\leq\delta\in(0,1]$ is exactly:
\[
F_2(\delta,T)=\frac12(T-1)(T-2)^2(T-3)\delta^{T-3}\left(\frac{1}{T-3}-\frac{2\delta}{(T-2)^2} -\frac{\delta^2}{T-1}+\frac{2\delta\log\delta}{T-2} \right),
\]
and similar formulas can be computed for the probability density functions $F_M(\delta,T)$ for higher values of $M$ (the case $M=1$ which yields $F_1(\delta,T)=\delta^{T-1}$ is quite trivial but it also follows from our approach).
\end{cor}
See the complete result in Corollary
\ref{cor:pdfdet}. As an illustrative example, Figs. \ref{Fig:PDFs} and \ref{Fig:CDFs} depict, respectively, the computed pdf and cdf of the chordal product for different values of $T$ and $M$.

\begin{figure}
    \centering
\includegraphics[width=.60\textwidth]{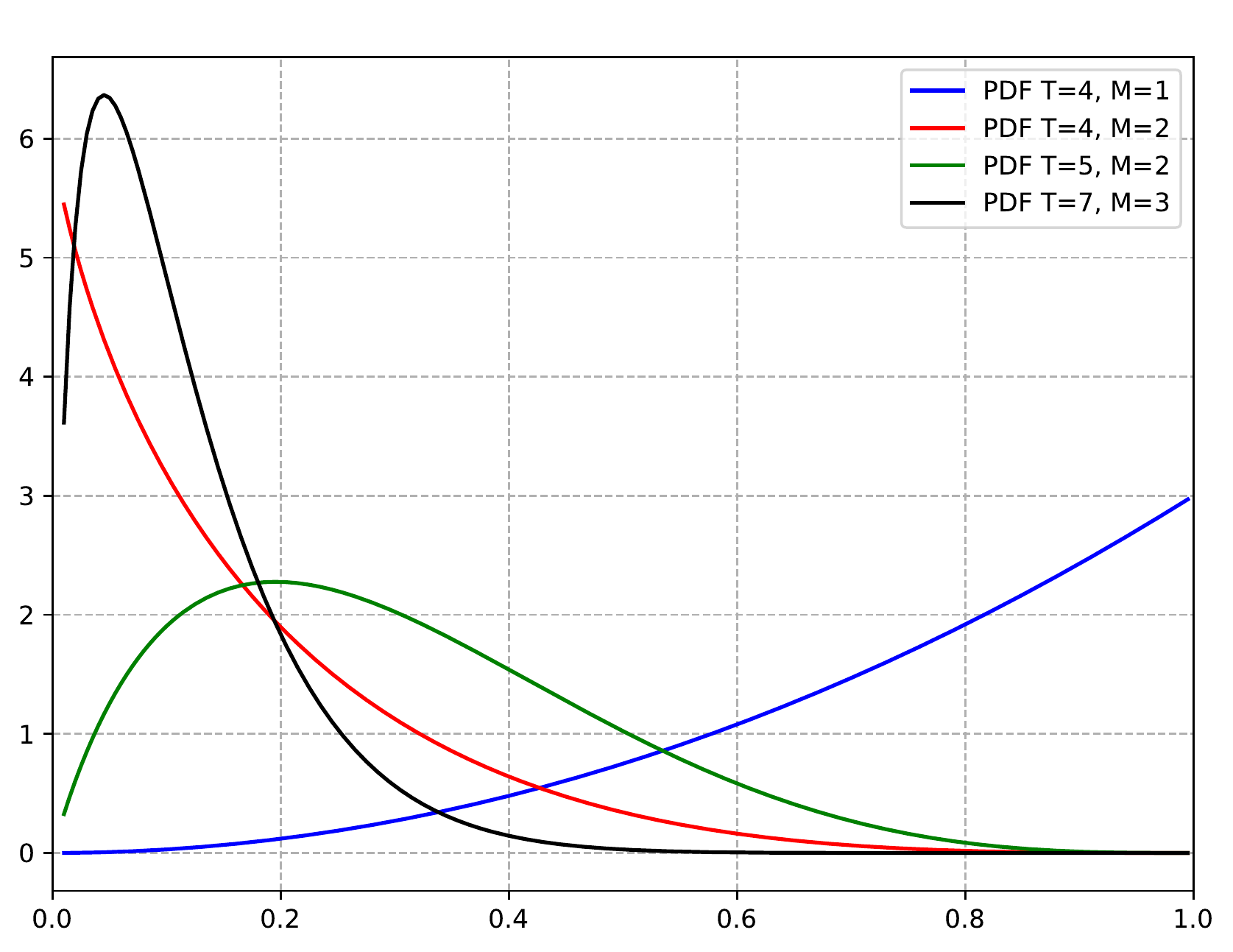}
     \caption{Probability density functions of $\det(\I_M-\B^H \A\A^H \B)$, when $[\B]\in\Gras$ is fixed and $[\A]\in\Gras$ is uniformly distributed on the Grassmannian.}
	\label{Fig:PDFs}
\end{figure}

\begin{figure}
    \centering
\includegraphics[width=.60\textwidth]{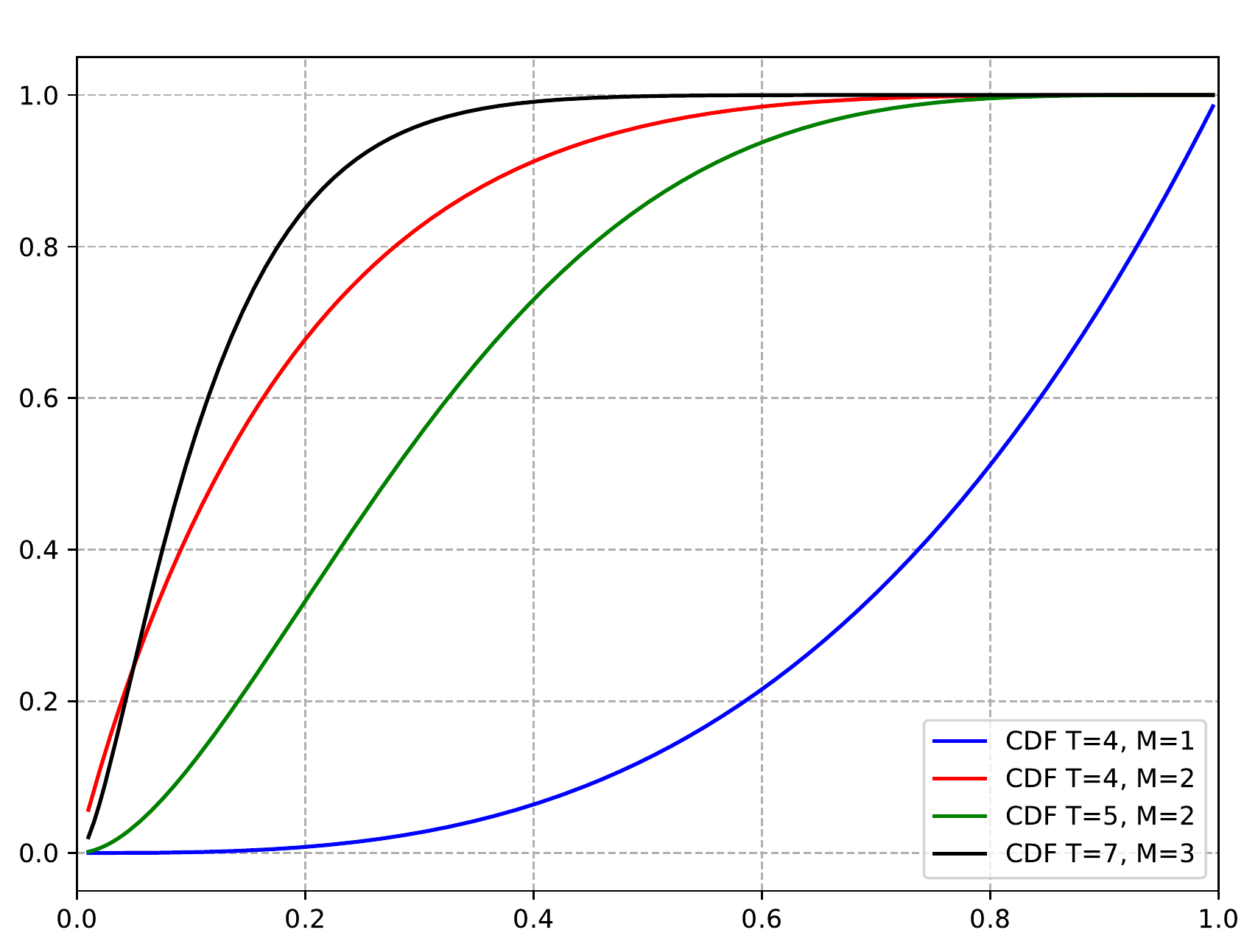}
     \caption{Cumulative distribution function of $\det(\I_M-\B^H \A\A^H \B)$, when $[\B]\in\Gras$ is fixed and $[\A]\in\Gras$ is uniformly distributed on the Grassmannian.}
	\label{Fig:CDFs}
\end{figure}
 
Using the statistical characterization above, we have derived a lower bound
on the number of elements in any code {in the Grassmannian} with a given minimum value of chordal product $\delta$. Following \cite{Barg_TIT02}, we call this result a Gilbert-Varshamov bound since its proof mimics the argument of that classical result.
\begin{cor}[Gilbert--Varshamov lower bound]\label{cor:lowerbounddet}Assume that $T\geq2M$. 
For any fixed $K\geq2$, there exists a code $[\X_1],\ldots,[\X_K]\in\Gras$ such that $\det(\I_M-\X_i^H\X_j\X_j^H\X_i)\geq\delta$ where $\delta$ is the unique solution of the equation:
  $$
  F_M(\delta;T)= \frac{1}{K};\text{ that is }\delta=F_M^{-1}\left(\frac{1}{K};T\right)
  $$
  Equivalently, given $\delta\in(0,1)$, there exists a code consisting of $K\geq\frac{1}{F_M(\delta;T)}$ elements and satisfying $\det(\I_M-\X_i^H\X_j\X_j^H\X_i)\geq\delta$ for $i\neq j$.
  \end{cor}
    \begin{exmpl}
  In the case $M=1$ we have that
  \[
  \det(\I_M- \X_i^H\X_j \X_j^H\X_i) = \sin^2 \theta,
  \]
  where $\theta\in[0,\pi/2]$ is the principal angle between the one-dimensional subspaces $[\X_i]$ and $[\X_j]$ in $\mathbb{G}(1,\mathbb{C}^T)$. That is to say, the chordal product coincides with the squared chordal distance. For uniformly distributed subspaces, the squared sine of the pairwise principal angle has cdf $F_1(\delta,T) = \delta^{T-1}$. The Gilbert-Varshamov bound shows that, for $\delta\in(0,1)$, there exist codes with cardinality $K$ and minimum chordal product $\delta = \sin^2 \theta$ such that
  \begin{equation}
      K > \delta^{-(T-1)} = \left( \sin \theta \right)^{-2(T-1)}.
  \end{equation}
  \end{exmpl}
\begin{exmpl}
   Let us now take $T=10, M=2$. Assume that we want to allocate $K=2^{9}$ points in $\Gras$. Then, Corollary \ref{cor:lowerbounddet} says that there exists a code $[\X_1],\ldots,[\X_K]$ such that for all these points the chordal product is at least $\delta$, the unique solution of:
  $$
  \frac12(T-1)(T-2)^2(T-3)\delta^{T-3}\left(\frac{1}{T-3}-\frac{2\delta}{(T-2)^2} -\frac{\delta^2}{T-1}+\frac{2\delta\log\delta}{T-2} \right)=\frac1{K},
  $$
  that is
  $$
  \log_2\left(2016\delta^7\left(\frac17-\frac{2\delta}{64}-\frac{\delta^2}{9}+ \frac{\delta\log\delta}{4}\right)\right)=-9,
  $$
  which yields $\delta\approx0.2129$. The numerical algorithm in \cite{CuevasTCOM} produces in this case $[\X_1],\ldots,[\X_{2^9}]$ with minimum determinantal value $0.3958>0.2129$.
  \end{exmpl}
  

\subsection{Historical discussion}\label{sec:history}
There exist several results on packings on Grassmannian spaces but they are rather centered in finding codes such that the mutual chordal distance between different elements $[\X_i],[\X_j]$ is close to maximal. For example, in \cite{rankin} we find bounds for the mutual distance of any code with a fixed number of elements (this is known as the Rankin bound). Gilbert--Varshamov bounds have also been obtained for that chordal distance by resorting to calculations of the volume of a metric ball of radius $\delta$ in $\Gras$, see \cite{Barg_TIT02} and \cite{Dai_TIT08}. The case that $\delta$ is sufficiently small was analyzed in \cite{Dai_TIT08}, \cite{Henkel2005}, and the real case has also been studied, see \cite{Conway96} and references therein. But to our knowledge our results are the first theoretical bounds on codes focusing on the explicit use of the chordal product, which is the key figure of merit in non--coherent communications.

In the case $M=1$ the Grassmannian becomes the projective space, the chordal product equals the squared chordal distance, and the literature is much more prolific, going back to \cite{Shannon_59} (although Shannon studied the case of the geodesic, not chordal, distance), \cite{Barg_TIT02} for the chordal and the geodesic distance and more recently \cite{jasper2019game} where a more complete set of references can be found. The optimal value of \eqref{eq:unionbound2} in that case has been studied in \cite{BE2018} and \cite{Anderson} as a case of Riesz energy, showing that the minimum value is equal to the average computed in Corollary \ref{cor:optimalrandom}, minus a term of the form
\[
O\left(K^{1+\frac{N}{T-1}}\right)=o(K^2),\quad\text{ for }N\leq T-2.
\]

\section{Proof of Theorem  \ref{th:expected}}\label{sec:statcharact}

First assume that $p$ is an integer in the range of the hypotheses. Let ${\mathrm E}(p,T)$ be the expected value in the theorem (we omit the dependence on $M$ in the notation). By unitary invariance, we can assume that $\B=\binom{\I_M} {\bf 0}$. If we write the expected value using Proposition \ref{prop:integrales} and we pass to polar coordinates we get
\begin{align*}
{\rm E}(p,T)=&C(T)\int_{\tilde{\A}\in\mathbb{C}^{(T-M)\times M}}\frac{\det\left(\I_M-(\I_M+\tilde \A^H\tilde \A)^{-1}\right)^p}{\det(\I_M+\tilde \A^H\tilde \A)^{T}}\,d\tilde{\A}\\
&=C(T)\int_{\tilde{\A}\in\mathbb{C}^{(T-M)\times M}}\frac{\det\left( \A^H\tilde \A\right)^p}{\det(\I_M+\tilde \A^H\tilde \A)^{T+p}}\,d\tilde{\A}\\
=&C(T)\int_0^\infty\rho^{2M(T-M)+2p M-1}\int_{\underset{\|\tilde{\A}\|_F=1}{\tilde \A\in\mathbb{C}^{(T-M)\times M}}}\frac{\det(\tilde \A^H\tilde \A)^p }{\det(\I_M+\rho^2\tilde \A^H\tilde \A)^{T+p}}\,d\tilde{\A}\,d\rho,
\end{align*}
where we omit the dependence on $M$ in the constant:
$$
C(T)=\frac{1}{Vol(\Gras)}\stackrel{Lemma~\ref{lem:volumenG}}{=}\frac{(T-M)!\cdots(T-1)!}{\pi^{M(T-M)}1!\cdots(M-1)!}.
$$
Since the integrand of the inner integral depends only on the singular values of $\tilde \A$, we can take it to the set $\mathbb S_M^+$ consisting of ordered tuples of positive numbers $\sigma_1>\ldots>\sigma_M$ with the property that $\sigma_1^2+\cdots+\sigma_M^2=1$, see for example \cite[Th. 3.3]{IMAJNA}, that yields
\begin{multline*}
{\rm E}(p,T)=D(T)\times\\\int_0^\infty\rho^{2M(T-M)+2p M-1} \int_{\mathbb S_M^+}\frac{(\sigma_1\cdots\sigma_M)^{2p+2T-4M+1}\prod_{j\neq k}(\sigma_k^2-\sigma_j^2)^2 }{\prod_{m=1}^M(1+\rho^2\sigma_m^2)^{T+p}}\,d\sigma_1\cdots d\sigma_M\,d\rho,
\end{multline*}
where
\begin{align*}
  D(T) =& \frac{C(T)Vol(\CU_{T-M})Vol(\CU_M)}{Vol(\CU_{T-2M})2^{M(T-M)}\pi^M}.
\end{align*}
It follows immediately that ${\rm E}(p,T)/D(T)= {\rm E}(0,T+p)/D(T+p)=1/D(T+p)$, that is,
\begin{align*}
   {\rm E}(p,T)=& \frac{D(T)}{D(T+p)} \\
  =& \frac{C(T)Vol(\CU_{T-M}) Vol(\CU_M)}{2^{M(T-M)}Vol(\CU_{T-2M})}\frac{2^{M(T+p-M)}Vol(\CU_{T+p-2M})}{C(T+p) Vol(\CU_{T+p-M}) Vol(\CU_M)}
   \\
  =&  \frac{(2\pi)^{Mp}(T-M)!\cdots(T-1)!Vol(\CU_{T-M}) Vol(\CU_{T+p-2M})}{Vol(\CU_{T-2M})(T+p-M)!\cdots(T+p-1)! Vol(\CU_{T+p-M})}.
\end{align*}
The volume of the unitary group is known (see \cite[p. 28]{IMAJNA}): 
$$
Vol(\CU_k)=\frac{(2\pi)^{k(k+1)/2}}{1!\cdots(k-1)!}.
$$
The theorem (for integer $p$ in the range) follows by substituting the known values in the constants above.

On the other hand, it is known that the $p$th moment of a beta distributed random variable with parameters $\alpha>0$ and $\beta >0$ denoted as $x \sim {\rm Beta}(\alpha,\beta)$ is \cite{Srivastavabook} 
\begin{equation}
\label{eq:betamoments}
{\rm E}[x^p] = \frac{\Gamma(\alpha+\beta)\Gamma(\alpha+p)}{\Gamma(\alpha+\beta+p)\Gamma(\alpha)},
\end{equation}
so the $m$th product term in \eqref{eq:detmoment} corresponds to the $p$th moment of a beta distributed random variable with parameters $\alpha_m = T-M+1-m$ and $\beta_m = M$, thus proving that the distribution of $\det(\I_M-\B^H \A\A^H \B)$ is equivalent to the distribution of the product of $M$ independent beta random variables (this is an instance of the Hausdorff moments problem, hence the distribution is uniquely determined by its moments). Notice that \eqref{eq:betamoments} is valid for $p+ \alpha >0$, and since $m$ can get up to $M$ this entails to $p>2M-T-1$. Now that we have characterized  $\det(\I_M-\B^H \A\A^H \B)$ as a product of beta distributed random variables, we can write down the formula for its moments for noninteger $p>2M-T-1$, finishing the proof of the theorem.





\section{Probability density function of the chordal product}\label{sec:pdfcdf}

Can we effectively recover the pdf of the random variable $x= \det(\I_M-\B^H \A\A^H \B)$ from its moments? If the density function is $f(x)$ and the moments are $\mathcal M_n$ then we have the classical formula:
\begin{equation}\label{eq:pdffrommoments}
  f(x)=\int_{-\infty}^\infty e^{2i\pi xs}\sum_{n=0}^\infty\frac{(-2i\pi s)^n}{n!}\mathcal M_n\,ds.
\end{equation}
Following \cite[Th. 7]{SpringerThomson} a closed-form expression for the pdf can actually be written down in terms of certain special functions called Meijer $G$--functions. However, this expression is quite involved and requires extra work in practice for the derivation of bounds. In the following, we show that we can obtain simpler closed-form formulas for small values of $M =1,2,3$. They represent the most practical use cases in noncoherent communications. Moreover, we also provide a general recursive procedure to obtain the pdfs for higher values:
\begin{cor}\label{cor:pdfdet}
 Let $T\geq2M$. The probability density function (pdf) of $\det(\I_M-\B^H \A\A^H \B)$, when $[\B]\in\Gras$ is fixed and $[\A]\in\Gras$ is uniformly distributed on the Grassmannian, for $M = 1,2,3$ is:
 \begin{align*}
 M=1 \, \to \,  f_{1}(x;T)=&(T-1)x^{T-2},\\
  M=2 \, \to \,  f_{2}(x;T)=&\frac12(T-1)(T-2)^2(T-3)x^{T-4}\left(1-x^2+2x\log x\right),\\
 M=3 \, \to \,  f_{3}(x;T)=&\frac1{288}(T-1)(T-2)^2(T-3)^3(T-4)^2(T-5)x^{T-6}\times\\
 &\left(1+80x-162x^2+80x^3+x^4+24x\log x-24x^3\log x-36x^2\log^2x\right)
 \end{align*}
The cumulative distribution function (cdf) $F_M(x,T)=\int_0^xf_M(s,T)\,ds$ for these three cases is respectively:
 \begin{align*}
 M=1\to&F_{1}(x;T)=x^{T-1}\\
 M=2\to&F_{2}(x;T)=\frac12(T-1)(T-2)^2(T-3)x^{T-3}\left(\frac{1}{T-3}-\frac{2x}{(T-2)^2} -\frac{x^2}{T-1}+\frac{2x\log x}{T-2} \right)\\
 M=3\to&F_{3}(x;T)=\frac1{288}(T-1)(T-2)^2(T-3)^3(T-4)^2(T-5)x^{T-5}\times Q,
  \end{align*}
  with
  \begin{multline*}
    Q=\frac{1}{T-5}
 +\frac{80x}{T-4}
 -\frac{24x}{(T-4)^2}
 -\frac{162x^2}{T-3}
 -\frac{72x^2}{(T-3)^3}
 +\frac{80x^3}{T-2}\\
 +\frac{24x^3}{(T-2)^2}
 +\frac{x^4}{T-1}
 +\frac{24x\log x}{T-4}
 +\frac{72x^2\log x}{(T-3)^2}
 -\frac{24x^3\log x}{T-2}
 -\frac{36x^2\log^2x}{T-3}.
  \end{multline*}
  For arbitrary higher values of $M$ the pdf has the form:
  \begin{align*}
     & f_M(x;T) = (T-M)^M\prod_{m=1}^{M-1}(T-m)^m(T-M-m)^{M-m}\cdot\\
      &\left[
      \sum_{m=1}^M\sum_{l=1}^m\frac{A_{ml}(-1)^{l-1}}{(l-1)!}x^{T-m-1}\log^{l-1}x + 
      \sum_{m=1}^{M-1}\sum_{l=1}^{M-m}\frac{B_{ml}(-1)^{l-1}}{(l-1)!}x^{T-m-M-1}\log^{l-1}x
      \right]
  \end{align*}
  where the $M^2$ coefficients $A_{ml}, B_{ml}$ can be obtained (e.g. with the aid of symbolic computation software) by solving the linear system of $M^2$ equations resulting from equating coefficients on both sides for the polynomial identity:
\begin{align*}
&\sum_{m=1}^M\sum_{l=1}^m A_{ml}(x-m)^{m-l}\prod_{i\neq m}^M (x-i)^i\prod_{i=1}^{M-1}(x-M-i)^{M-i} + \\
& +\sum_{m=1}^{M-1}\sum_{l=1}^{M-m} B_{ml}(x-M-m)^{M-m-l}\prod_{i=1}^M (x-i)^i\prod_{i\neq m}^{M-1}(x-M-i)^{M-i} = 1.
\end{align*}
 \end{cor}
 
 \begin{proof}

For $M=1$ and integer $p\geq0$ note that
$$
\int_0^1x^p \underbrace{(T-1)x^{T-2}}_{f_{1}(x;T)}\,dx= \frac{T-1}{T+p-1},
$$
and hence the claimed pdf satisfies Theorem \ref{th:expected} and must be the searched distribution. With the help of some integral formulas for the $\log$ function it is easy to check that
$$
\int_0^1x^pf_{2}(x;T)\,dx=\frac{(T-1)(T-2)^2(T-3)}{(T+p-1)(T+p-2)^2(T+p-3)},
$$
which again satisfies Theorem \ref{th:expected} and we are done. A more lengthy but trivial computation gives the case $M=3$.

These formulas and the general case for higher values of $M$ can be derived from the following procedure. The moments of the chordal product determinant from Theorem \ref{th:expected} are
$$
\mathcal{M}_p(T,M) = \prod_{m=1}^M\frac{(T-m)!(T+p-m-M)!}{(T-m-M)!(T+p-m)!}. 
$$
By expanding the factorials and collecting terms in the product, this can be rewritten as
$$
\mathcal{M}_p(T,M) = \prod_{m=1}^M\left(\frac{T-m}{T+p-m}\right)^m\;\cdot\;\prod_{m=1}^{M-1}\left(\frac{T-M-m}{T+p-M-m}\right)^{M-m}
$$
so the denominator $D$ is the value at $x=T+p$ of the polynomial
$$
D(x) = \prod_{m=1}^M(x-m)^m\prod_{m=1}^{M-1}(x-M-m)^{M-m}.
$$
Notice that this is a product of all-different real root factors $(x-\alpha)$ with varying multiplicities, so its inverse has a partial fraction decomposition
$$
\frac{1}{D} = \sum_{m=1}^M\sum_{l=1}^m\frac{A_{ml}}{(x-m)^l} + \sum_{m=1}^{M-1}\sum_{l=1}^{M-m}\frac{B_{ml}}{(x-M-m)^l},
$$
for some coefficients $A_{ml},\, B_{ml}\in\mathbb{R}$. Following one of the usual procedures to solve for these coefficients, for general $x$, multiplying by $D(x)$ on both sides yields the polynomial equation of order $M^2-1$:
\begin{align*}
&\sum_{m=1}^M\sum_{l=1}^m A_{ml}(x-m)^{m-l}\prod_{i\neq m}^M (x-i)^i\prod_{i=1}^{M-1}(x-M-i)^{M-i} + \\
& +\sum_{m=1}^{M-1}\sum_{l=1}^{M-m} B_{ml}(x-M-m)^{M-m-l}\prod_{i=1}^M (x-i)^i\prod_{i\neq m}^{M-1}(x-M-i)^{M-i} = 1.
\end{align*}
Expanding and gathering terms by powers of $x$, one can equate the coefficient of $x^0$ to $1$ and the coefficients of $x^n$, for $n=1,\dots, M^2-1$, to $0$ to obtain a linear system of $M^2$ equations in the $M^2$ coefficients $A_{ml},\, B_{ml}$, and solve for them.

Now, notice that by the Laplace transform properties for $a,b$ nonnegative integers
$$
\int_0^1 x^a\log^b x\,dx = (-1)^b\mathcal{L}[t^b](a+1) = (-1)^b\frac{b!}{(a+1)^{b+1}},
$$
and thus
$$
\frac{1}{(T+p-m)^l}=\frac{(-1)^{l-1}}{(l-1)!}\int_0^1 x^p x^{T-m-1}\log^{l-1}x\,dx.
$$
Hence, by expressing every term of the partial fraction decomposition in this integral form, the function $f_M(x;T)$ can be identified inside the moment function written as an integral:
\begin{align*}
  \mathcal{M}_p(T,M) = & (T-M)^M\prod_{m=1}^{M-1}(T-m)^m(T-M-m)^{M-m}\cdot\\
  & \cdot\int_0^1 x^p\left[
  \sum_{m=1}^M\sum_{l=1}^m\frac{A_{ml}(-1)^{l-1}}{(l-1)!}x^{T-m-1}\log^{l-1}x + \right.\\
  &\left. \sum_{m=1}^{M-1}\sum_{l=1}^{M-m}\frac{B_{ml}(-1)^{l-1}}{(l-1)!}x^{T-m-M-1}\log^{l-1}x
  \right] dx,
\end{align*}
where the factors in front of the integral all come from the numerator over $D$ in $\mathcal{M}_h(T,M)$. This finishes the proof of the corollary.
 \end{proof}
 

\section{Proof of Corollary \ref{cor:lowerbounddet}} \label{sec:GilbertVarshamov}

  Let $\delta$ be the unique solution of the equation $F_M(\delta;T)=K^{-1}$. Let $G$ be the maximum number of points in $\Gras$ that can be allocated with the claimed property. We must prove that $G\geq K$. Indeed, assume that $G<K$ an let $\X_1,\ldots,\X_G$ be a code with $\det(\I_M-\X_i^H\X_j\X_j^H\X_i)\geq\delta$ for all $1\leq i,j\leq G$. We note that
  \begin{multline*}
    \frac{1}{Vol(\Gras)}Vol\left(\cup_{i=1}^G\{[\A]\in\Gras:\det(\I_M-\X_i^H\A\A^H\X_i)\leq\delta\}\right)\leq\\\frac{1}{Vol(\Gras)}\sum_{i=1}^G Vol\left(\{[\A]\in\Gras:\det(\I_M-\X_i^H\A\A^H\X_i)\leq\delta\}\right)=\\
    GF_M(\delta;T)=\frac{G}{K}<1,
\end{multline*}
  and we thus deduce that there exists $[\A]\in\Gras$ such that
  $$
  \det(\I_M-\X_i^H\A\A^H\X_i)>\delta\quad \forall\;1\leq i\leq G.
  $$
  But then the code $\X_1,\ldots,\X_G,\X_{G+1}$ with $\X_{G+1}=\A$ also satisfies the claimed property and has $G+1$ points, which contradicts the definition of $G$.

\vspace{6pt} 

\appendix

\section[Alternative parameterization of $\Gras$ ]{Alternative parameterization of the Grassmannian and the density function of $\tilde \A$ in $\binom{\I_M}{\tilde \A}\in\Gras$}\label{sec:parametrization}
We recall the volume of the Grassmannian for completeness.
\begin{lem}\label{lem:volumenG}
The volume of the Grassmannian $\Gras$ is:
$$
Vol(\Gras)=\frac{\pi^{M(T-M)}1!\cdot 2!\cdots(M-1)!}{(T-M)!\cdot (T-M+1)!\cdots(T-1)!}
$$
\end{lem}
\begin{proof}
  This is a classical fact: since the Grassmannian is formally defined as a quotient of the Stiefel manifold $\St$ (i.e. the set of $T\times M$ complex matrices $\X$ such that $\X^H\X=\I_M$) by the unitary group $\mathcal U_M$, the volume of $\Gras$ is the quotient of the volumes of the Stiefel and unitary matrices which is well--known, see for example \cite{Hua} (note that there exist several normalizations for the Riemannian structure of the classical groups, leading to different volume formulas. We use the standard that considers $\St$ and $\mathcal U_M$ as submanifolds of their ambient affine spaces, with the inherited structure).
\end{proof}
Recall that a $p\times n$ complex matrix $\X$ is distributed as a complex matrix-variate $t$ distribution with $\nu$ degrees of freedom when its density is given by
\begin{equation}
    p(\X) = C^{-1} \det(\I_p + \X^H\X)^{-(\nu +p+n-1)}.
\end{equation}
\begin{prop}\label{prop:integrales}
Let $T\geq 2M$. If $[\A]$ is uniformly distributed in $\Gras$ and we write $[\A]=\begin{bmatrix}{\I_M}\\{\tilde \A}\end{bmatrix}$ (note that there exists a unique representative of that form), then $\tilde \A$ has density
\begin{equation}
\frac{1}{Vol(\Gras)\det(\I_M+\tilde \A^H\tilde \A)^{T}}.\label{eq:matrix_t_dist}
\end{equation}
Hence, $\tilde\A\in\mathbb C^{(T-M)\times M}$ follows a matrix--variate $t$ distribution with $\nu=1$ degrees of freedom.
In other words, for any measurable non--negative or integrable function $f:\Gras\to\mathbb{C}$,
\begin{align*}
\int_{[\A]\in\Gras}f([\A])\,d[\A]=&\int_{\tilde{\A}\in\mathbb{C}^{(T-M)\times M}}\frac{f\binom{\I_M}{\tilde \A}}{\det(\I_M+\tilde \A^H\tilde \A)^{T}}\,d\tilde{\A}
\\
=&\int_{\tilde{\A}\in\mathbb{C}^{(T-M)\times M}}\frac{f\binom{(\I_M+\tilde\A^H\tilde\A)^{-1/2}}{\tilde \A(\I_M+\tilde\A^H\tilde\A)^{-1/2}}}{\det(\I_M+\tilde \A^H\tilde \A)^{T}}\,d\tilde{\A}.
\end{align*}
\end{prop}
\begin{proof}
  This result has been proved in \cite[Prop. 1 and Cor. 1]{GrassLattice}, by showing that both sides of the equality are equal to
  \[
  \int_{\W\in \mathbb C^{(T-M)\times M}, \|\W\|_{op}<1}f\left(\begin{bmatrix}\sqrt{\I_M-\W^H\W}\\\W\end{bmatrix}\right)\,d\W,
  \]
  with $\|\cdot\|_{op}$ the operator norm. Note that the two integrals on $\C^{(T-M)\times M}$ are equal since $f$ is a function defined in the Grassmannian and hence its value is independent of the choice of representatives. Moreover, the advantage of the last expression in the proposition is that 
  $$
  \X=\binom{(\I_M+\tilde\A^H\tilde\A)^{-1/2}}{\tilde \A(\I_M+\tilde\A^H\tilde\A)^{-1/2}}
  $$
  is a Stiefel matrix, i. e. it satisfies $\X^H\X=\I_M$.
\end{proof}

\bibliographystyle{alpha}
\addcontentsline{toc}{section}{References}
\bibliography{main}

\Addresses

\end{document}